\documentclass[10pt,twocolumn]{article}

\usepackage{amsmath,amssymb}
\usepackage{amsthm}
\usepackage{enumitem}
\usepackage{geometry}
\usepackage{hyperref}
\usepackage{listings}
\usepackage{times}

\theoremstyle{plain}
\newtheorem{theorem}{Theorem}[section]
\newtheorem{lemma}[theorem]{Lemma}
\newtheorem{proposition}[theorem]{Proposition}

\theoremstyle{definition}

\theoremstyle{remark}

\title{\vspace{-1em}Many Scores Force Many Reads:\\
An Optimality Proof for Score of Parentheses}
\author{Andrew Au\\[0.2em]
\small Independent Researcher\\[0.2em]
\small Corresponding author: \texttt{cshung@gmail.com}}
\date{\vspace{-1.5em}}

\begin{document}

\twocolumn[
\maketitle
\vspace{-1em}
]

\begin{abstract}
LeetCode Problem 856, \emph{Score of Parentheses}, admits a one-pass
algorithm using \(O(n)\) time and \(O(1)\) auxiliary registers.  Its
correctness follows from a local interpretation: every occurrence of
\texttt{()} contributes a power of two determined by its nesting depth.
The space bound is immediately optimal in the usual programming-challenge
model, but the linear-time bound is less obvious.  In particular, a direct
adversary that forces an algorithm to inspect every character is difficult
to maintain under the promise that the input is balanced.

We instead prove a counting lower bound.  For every \(k\), we construct
\(2^k\) balanced parenthesis strings of the same length \(4k+2\), all
having distinct scores.  The construction encodes each bit by one of two
operations, changing a current score \(v\) to either \(2v+1\) or \(2v+2\).
Consequently, any deterministic character-query decision tree computing
the score needs depth at least \(k\).  A two-character padding argument
extends the result to every even input length, proving an
\(\Omega(n)\) worst-case time lower bound.  The one-pass algorithm is
therefore asymptotically optimal in the unit-cost RAM model.
\end{abstract}


\section{Introduction}

A balanced parenthesis string is generated recursively from the empty
string by concatenation and enclosure.  The \emph{score of parentheses}
is defined by the rules
\begin{align*}
\operatorname{score}(\texttt{()})&=1,\\
\operatorname{score}(AB)&=\operatorname{score}(A)+\operatorname{score}(B),\\
\operatorname{score}(\texttt{(}A\texttt{)})
  &=2\operatorname{score}(A),
\end{align*}
where \(A\) and \(B\) are nonempty balanced strings.  LeetCode
Problem 856 asks for the score of a given balanced string \(s\)
\cite{LeetCode856}.

A standard stack solution mirrors the recursive definition and uses
\(O(n)\) auxiliary space.  A more economical observation eliminates the
stack.  Each adjacent pair \texttt{()} is an indivisible unit of score
one, and every surrounding pair of parentheses doubles its contribution.
It is therefore enough to track the current power of two and add it
whenever an adjacent pair closes.

This gives a one-pass \(O(n)\)-time, \(O(1)\)-register algorithm.  The
space bound leaves little room for improvement.  The interesting question
is whether the time bound is also optimal.

The informal claim that ``the algorithm must read the input'' is not by
itself a proof.  A correct algorithm need not necessarily inspect every
position, and the promise that the string is balanced allows some
characters to be inferred from others.  Rather than force every position
to be inspected, we prove the weaker but asymptotically sufficient fact
that a linear number of binary observations is necessary.

The proof uses the number of possible outputs.  Counting outputs alone is
not generally a time lower bound; it becomes one only after the input
access model is specified.  In the character-query model, each probe has
only two possible outcomes.  A depth-\(q\) decision tree can therefore
distinguish at most \(2^q\) output values.  Our main construction supplies
exponentially many distinct scores among inputs of one fixed length.

\section{A Constant-Space Algorithm}

The following implementation is a slightly explicit version of the
one-pass algorithm.  Initializing \(f\) is redundant for a valid nonempty
input, whose first character is necessarily \texttt{(}, but makes its
meaning visible from the start.

\par\noindent\begin{minipage}{\columnwidth}
\begin{lstlisting}[language=Python]
class Solution:
    def scoreOfParentheses(self, s: str) -> int:
        ans = 0
        m = 1
        f = 0

        for c in s:
            if c == '(':
                m *= 2
                f = 1
            else:
                m //= 2
                ans += m * f
                f = 0

        return ans
\end{lstlisting}
\end{minipage}\par

The variable \(m\) records a power of two associated with the current
depth.  The flag \(f\) records whether the current character closes an
adjacent pair \texttt{()}.

\subsection{Atomic pairs}

For an occurrence of \texttt{()} in a balanced string, define its
\emph{enclosing depth} to be the number of matched parenthesis pairs that
strictly contain it.  For example, the three atomic pairs in
\texttt{(()())()} have enclosing depths \(1,1\), and \(0\), respectively.

\begin{lemma}
\label{lem:atomic}
For every nonempty balanced parenthesis string \(S\),
\[
\operatorname{score}(S)
  =\sum_{p}2^{d(p)},
\]
where the sum ranges over all occurrences \(p\) of \texttt{()} in \(S\)
and \(d(p)\) is the enclosing depth of \(p\).
\end{lemma}

\begin{proof}
Proceed by structural induction.

For \(S=\texttt{()}\), there is one atomic pair at depth zero, so the sum
is \(2^0=1\).

If \(S=AB\), the atomic pairs of \(A\) and \(B\) retain their enclosing
depths.  The proposed sum is therefore the sum of the corresponding
expressions for \(A\) and \(B\), matching
\(\operatorname{score}(AB)=\operatorname{score}(A)+
\operatorname{score}(B)\).

If \(S=\texttt{(}A\texttt{)}\), every atomic pair of \(A\) gains one
enclosing pair.  Every term \(2^{d(p)}\) is consequently doubled,
matching
\(\operatorname{score}(\texttt{(}A\texttt{)})
=2\operatorname{score}(A)\).
\end{proof}

\begin{theorem}
\label{thm:algorithm}
The algorithm returns the score of \(s\) in \(O(n)\) time using \(O(1)\)
auxiliary registers.
\end{theorem}

\begin{proof}
After any processed prefix, let \(d\) be the number of unmatched opening
parentheses.  The updates to \(m\) maintain \(m=2^d\): an opening
parenthesis increments \(d\) and doubles \(m\), while a closing
parenthesis decrements \(d\) and halves \(m\).

The flag \(f\) equals one exactly when the previously processed character
was \texttt{(}.  Thus, on processing a closing parenthesis, \(f=1\)
exactly when that character completes an occurrence of \texttt{()}.
After \(m\) is halved, its value is \(2^d\), where \(d\) is precisely the
number of pairs enclosing that occurrence.  The algorithm therefore adds
the term assigned to this atomic pair by Lemma~\ref{lem:atomic}.  It adds
nothing on any other closing parenthesis.  Hence \texttt{ans} is the sum
in Lemma~\ref{lem:atomic} when the scan ends.

The loop performs constant work for each of the \(n\) characters and
retains only \texttt{ans}, \(m\), \(f\), and the loop variable.
\end{proof}

\section{Computational Model}
\label{sec:model}

We use the unit-cost random-access model customary in
programming-challenge analysis.  The input length is known.  The input is
stored as an array with one parenthesis per addressable cell, and probing
one cell costs one unit of time.  A probe returns one of the two symbols
\texttt{(} and \texttt{)}.  The algorithm may choose probes adaptively,
revisit positions, and perform arbitrary unit-cost arithmetic and control
operations between probes.  Integer registers are assumed large enough
to hold the score.

The algorithm is deterministic and must return the exact score for every
balanced input of the specified length.  This is the deterministic
character-query model, represented by a binary decision tree
\cite{CLRS,AroraBarak}.  Every internal node specifies a position to
probe, its two outgoing edges represent the observed character, and each
leaf carries the returned score.

This model deliberately excludes the cost of representing a large
integer.  In a bit-cost model, a linear lower bound is immediate: a
fully nested string of length \(n\) has score
\[
2^{n/2-1},
\]
whose binary representation has \(n/2\) bits.  That observation is valid
but does not explain why the input itself contains a linear amount of
relevant information.  Our lower bound continues to hold when arithmetic
on the score is free.

\section{An Exponential Family of Scores}

We now construct many equal-length balanced strings with different
scores.  For any nonempty balanced string \(X\), define
\begin{align*}
F_0(X)&=\texttt{(}X\texttt{)()},\\
F_1(X)&=\texttt{(}X\texttt{())}.
\end{align*}
In words, \(F_0\) encloses \(X\) and then appends \texttt{()}, whereas
\(F_1\) first appends \texttt{()} and then encloses the result.  Both
operations add four characters and preserve balancedness.

\begin{lemma}
\label{lem:transform}
If \(X\) has score \(v\), then
\[
\operatorname{score}(F_0(X))=2v+1,
\qquad
\operatorname{score}(F_1(X))=2v+2.
\]
\end{lemma}

\begin{proof}
The defining score rules give
\[
\operatorname{score}(\texttt{(}X\texttt{)()})
  =2v+1
\]
and
\[
\operatorname{score}(\texttt{(}X\texttt{())})
  =2\operatorname{score}(X\texttt{()})
  =2(v+1).
\]
\end{proof}

Let \(b=b_1b_2\cdots b_k\in\{0,1\}^k\), with the bits listed in the
order in which they are processed.  Begin with
\[
X_{\epsilon}=\texttt{()},
\]
and recursively set
\[
X_{b_1\cdots b_i}=F_{b_i}(X_{b_1\cdots b_{i-1}}).
\]
The construction may equivalently be described as processing the bits of
an ordinary binary word from most significant to least significant.
Processing them in the reverse order merely bit-reverses the indexing of
the same family.

\begin{proposition}
\label{prop:family}
For every \(b\in\{0,1\}^k\), the string \(X_b\) is balanced, has length
\(4k+2\), and has score
\[
\operatorname{score}(X_b)
  =2^{k+1}-1+\sum_{i=1}^{k}b_i2^{k-i}.
\]
\end{proposition}

\begin{proof}
Balancedness follows inductively because both \(F_0\) and \(F_1\) map a
balanced string to a balanced string.  The initial length is two and each
of the \(k\) transformations adds four, giving \(4k+2\).

For the score formula, use induction on \(k\).  The empty word gives
\(\operatorname{score}(X_{\epsilon})=1=2^{1}-1\).  Suppose a prefix
\(b_1\cdots b_i\) has score
\[
2^{i+1}-1+\sum_{j=1}^{i}b_j2^{i-j}.
\]
Write this score as \(v_i\).  Appending the next construction bit
\(b_{i+1}\) applies Lemma~\ref{lem:transform}, giving
\begin{align*}
v_{i+1}
&=2v_i+1+b_{i+1}\\
&=2^{i+2}-1
  +\sum_{j=1}^{i}b_j2^{i+1-j}+b_{i+1},
\end{align*}
which is the required formula for \(i+1\).

The formula also makes injectivity explicit.  If \(X_b\) and \(X_c\)
have equal scores, subtracting \(2^{k+1}-1\) shows that the \(k\)-bit
words \(b\) and \(c\) represent the same integer.  Uniqueness of binary
representation gives \(b=c\).
\end{proof}

Consequently the scores are the consecutive values
\[
2^{k+1}-1,\,
2^{k+1},\,
\ldots,\,
3\cdot2^k-2.
\]
In particular, all \(2^k\) constructed strings have distinct scores.
The first levels are
\[
\begin{array}{c|c|c}
b & |X_b| & \operatorname{score}(X_b)\\ \hline
\epsilon & 2 & 1\\
0 & 6 & 3\\
1 & 6 & 4\\
00 & 10 & 7\\
01 & 10 & 8\\
10 & 10 & 9\\
11 & 10 & 10
\end{array}
\]

\section{The Decision-Tree Lower Bound}

\begin{lemma}
\label{lem:leaves}
If a deterministic algorithm returns \(R\) distinct output values on a
set of inputs, then its character-query decision tree has at least \(R\)
leaves and worst-case depth at least \(\lceil\log_2 R\rceil\).
\end{lemma}

\begin{proof}
A leaf carries one fixed output.  Inputs with different correct outputs
cannot terminate at the same leaf, so at least \(R\) leaves are required.
A binary tree of depth \(q\) has at most \(2^q\) leaves.  Therefore
\(2^q\geq R\), or \(q\geq\lceil\log_2R\rceil\).
\end{proof}

\begin{theorem}
\label{thm:lower}
For every even \(n\), any deterministic algorithm that computes the
score of every balanced parenthesis string of length \(n\) requires
\[
\left\lfloor\frac{n-2}{4}\right\rfloor
\]
character probes in the worst case.  Consequently, Score of Parentheses
has an \(\Omega(n)\) time lower bound in the model of
Section~\ref{sec:model}.
\end{theorem}

\begin{proof}
First suppose \(n=4k+2\).  Proposition~\ref{prop:family} gives \(2^k\)
valid inputs of length \(n\) with distinct outputs.  By
Lemma~\ref{lem:leaves}, the decision tree has depth at least \(k\).

Now suppose \(n=4k+4\).  Append one copy of \texttt{()} to every string
in the length-\((4k+2)\) family.  Every resulting string is balanced, has
length \(n\), and has its score increased by exactly one.  The scores
therefore remain distinct, so the same \(k\)-probe lower bound applies.

In both cases,
\[
k=\left\lfloor\frac{n-2}{4}\right\rfloor.
\]
Since every probe costs one unit of time, the lower bound is
\(\Omega(n)\).
\end{proof}

\section{Optimality}

The upper and lower bounds now match asymptotically.

\begin{theorem}
In the unit-cost RAM and character-access model of
Section~\ref{sec:model}, the deterministic time complexity of Score of
Parentheses is
\[
\Theta(n).
\]
The one-pass algorithm of Section~2 achieves this bound using \(O(1)\)
auxiliary registers.
\end{theorem}

\begin{proof}
Theorem~\ref{thm:algorithm} gives the \(O(n)\) upper bound and constant
auxiliary-register usage.  Theorem~\ref{thm:lower} gives the
\(\Omega(n)\) worst-case lower bound.
\end{proof}

\section{Conclusion}

Balancedness makes a direct ``change one unread character'' adversary
awkward, but the output space provides a cleaner route.  Two
four-character transformations encode a bit while updating the score by
\(v\mapsto2v+1\) or \(v\mapsto2v+2\).  After \(k\) transformations,
equal-length valid strings realize \(2^k\) different scores, forcing
depth \(k\) in a binary character-query tree.  Thus the input carries a
linear amount of relevant information even when arithmetic and
large-integer representation are free.

\end{document}